\documentclass[a4paper]{article}

\usepackage{amsmath}
\usepackage{amssymb}
\usepackage{amsthm}
\usepackage[left=1.5in,right=1.5in,top=1.5in,bottom=1.5in]{geometry}

\newcommand{\R}{\mathbb{R}}

\newcommand{\BPP}{\mathsf{BPP}}
\newcommand{\SZK}{\mathsf{SZK}}
\newcommand{\NP}{\mathsf{NP}}
\newcommand{\AM}{\mathsf{AM}}
\newcommand{\PSPACE}{\mathsf{PSPACE}}

\newcommand{\EXP}{\mathsf{EXP}}
\newcommand{\NEXP}{\mathsf{NEXP}}

\newcommand{\IP}{\mathsf{IP}}
\newcommand{\MIP}{\mathsf{MIP1}}
\newcommand{\PCP}{\mathsf{PCP}}

\DeclareMathOperator{\Opt}{Opt}

\DeclareMathOperator*{\E}{\mathbb{E}}
\DeclareMathOperator{\poly}{\mathsf{poly}}

\newcommand{\MaxLin}{\textsc{Max $3$-Lin-$2$}}

\newcommand{\bool}{\{0,1\}}

\theoremstyle{plain}
\newtheorem{theorem}{Theorem}[section]

\newtheorem{lemma}[theorem]{Lemma}
\newtheorem{proposition}[theorem]{Proposition}
\theoremstyle{definition}
\newtheorem{definition}[theorem]{Definition}
\theoremstyle{remark}

\newtheorem{fact}[theorem]{Fact}

\begin{document}

  \title{On the Power of Many One-Bit Provers}

  \author{
    Per Austrin\thanks{E-mail: \texttt{austrin@kth.se}.}\\
    Aalto University and
    KTH Royal Institute of Technology\\
    \and
    Johan H{\aa}stad\thanks{E-mail: \texttt{johanh@kth.se}. Funded by ERC advanced grant 226202.}\\
    KTH Royal Institute of Technology
    \and
    Rafael Pass\thanks{E-mail: \texttt{rafael@cs.cornell.edu}. Funded
      by NSF, DARPA, AFOSR, a Sloan Foundation Fellowship, and a Microsoft Research Fellowship.}\\
    Cornell University
  }

  \maketitle

 \begin{abstract}
We study the class of languages, denoted by
$\MIP[k, 1-\epsilon, s]$, which have $k$-prover games
where each prover just sends a
\emph{single} bit, with completeness $1-\epsilon$ and soundness error
$s$.  For the case that $k=1$ (i.e., for the
case of interactive proofs), Goldreich, Vadhan and Wigderson ({\em
  Computational Complexity'02}) demonstrate that $\SZK$ exactly
characterizes languages having 1-bit proof systems with``non-trivial''
soundness (i.e., $1/2 < s \leq 1-2\epsilon$). We demonstrate that for
the case that $k\geq 2$, 1-bit $k$-prover games exhibit a
significantly richer structure:
\begin{itemize}
\item (Folklore) When $s \leq \frac{1}{2^k} - \epsilon$, $\MIP[k, 1-\epsilon,
  s] = \BPP$;
\item When $\frac{1}{2^k}  + \epsilon \leq s < \frac{2}{2^k}
    -\epsilon$, $\MIP[k, 1-\epsilon, s] = \SZK$;
\item When $s \ge \frac{2}{2^k}  + \epsilon$, $\AM \subseteq \MIP[k,
 1-\epsilon, s]$;
\item For $s \le 0.62 k/2^k$ and sufficiently large $k$, $\MIP[k, 1-\epsilon, s] \subseteq \EXP$;
 \item For $s \ge 2k/2^{k}$, $\MIP[k,
  1, 1-\epsilon, s] = \NEXP$.
\end{itemize}
As such, 1-bit $k$-prover games yield a natural ``quantitative''
approach to relating complexity classes such as $\BPP$,$\SZK$,$\AM$, $\EXP$,
and $\NEXP$.  We leave open the question of whether a more
fine-grained hierarchy (between $\AM$ and $\NEXP$) can be established
for the case when $s \geq \frac{2}{2^k} + \epsilon$.
\end{abstract}

  \section{Introduction}
We study the expressiveness of $k$-prover games (introduced by Ben-Or,
Goldwasser, Kilian and Wigderson \cite{BGKW88}), where each prover
sends a \emph{single} bit. Let $\MIP[k, 1-\epsilon, s]$ denote the
class of languages having a $k$-prover game where each prover sends a
single bit, completeness $1-\epsilon$, and soundness error $s$.
Throughout the paper, we think of $k$ as a constant and $\epsilon$ as
an arbitrarily small constant.  Clearly, for a fixed $k$, as $s$ increases the
corresponding complexity class can only become larger. We are
interested in understanding to what extent the complexity class grows,
and whether the growth is ``smooth'' or if threshold phenomena occur.

When the soundness error is ``too small'', only trivial languages can
have such games.  In particular, provers sending random bits succeed
with probability at least $(1-\epsilon)2^{-k}$, placing the language
of any protocol with smaller soundness in $\BPP$.
  \begin{theorem}[Folklore, implicit in \cite{BGS98}]
    For every $k \ge 1$, $\epsilon > 0$, we have
    $$\MIP[k,1-\epsilon,1/2^k-\epsilon] = \BPP$$
  \end{theorem}

An interesting result by Goldreich, Vadhan and Wigderson \cite{GVW02}
shows that when $k=1$ (i.e., for interactive proofs
\cite{GMR89,BM88}), whenever the soundness is ``non-trivial'', then  
$\MIP[1, 1-\epsilon, s]$ characterizes $\SZK$, the class of languages having statistical zero-knowledge proofs.
We here focus on the case when $k\geq 2$. As we shall see, in this
setting, 1-bit $k$-prover games contains a richer variety of
complexity classes. We take a first step towards characterizing these classes.

Our first result is a simple generalization of the result of
\cite{GVW02}: we show that when $\frac{1}{2^k}  + \epsilon \leq s <
\frac{2}{2^k} - \epsilon$, then $\MIP[k, 1-\epsilon, s]$
characterizes $\SZK$.

  \begin{theorem}
    \label{thm:szk_range}
    For every $k \ge 2$, $\epsilon > 0$, and $1/2^k + \epsilon < s <
    2/2^k - \epsilon$, we have
    $$
    \MIP[k, 1-\epsilon, s] = \SZK.
    $$
  \end{theorem}

Our main result next shows that when the soundness becomes just
slightly higher than $2/2^k$, $\MIP$s appear to become significantly more
powerful; in particular, they contain all of $\AM$.
  \begin{theorem} [Main Theorem]
\label{main.thm}
    For every $k \ge 2$ and $\epsilon > 0$
    $$\AM \subseteq \MIP[k,1-\epsilon, 2/2^k+\epsilon]$$
  \end{theorem}
For instance, when $k = 2$, our $\MIP$ has soundness error
$\frac{1}{2}+\epsilon$.
This result should be compared to H\aa stad's 3-bit $\PCP$
\cite{Has01} that achieves the same soundness error. Since every 1-bit
$k$-prover game yields a $k$-bit $\PCP$, our $\MIP$ yields a 2-bit
$\PCP$ for $\AM$ with soundness error $1/2 + \epsilon$; in contrast,
the $\PCP$ resulting from our $\MIP$ is exponentially long, whereas
H\aa stad's $\PCP$ is polynomially long.  Nonetheless, as we shall see
shortly, our $\MIP$ construction heavily relies on H\aa stad's $\PCP$.

We leave open the question of whether $\MIP[k,1-\epsilon,
2/2^k+\epsilon]$ contains even richer complexity classes than $\AM$.
As a first step towards this question, we note that $\EXP$ is an
upper bound on this class.
 \begin{theorem}
   \label{thm:exp_bound}
   For all sufficiently large $k$, $\epsilon > 0$, $s \leq \frac{0.62k}{2^k}(1 - \epsilon)$ we have
   $$\MIP[k, 1-\epsilon, s] \subseteq \EXP.$$
   This holds also for $k = 3$ and $s \le 1/2-\epsilon$.
 \end{theorem}

Finally, we prove that for $k \ge 3$ and sufficiently high soundness
error, $k$-prover 1-bit $\MIP$s capture all of $\NEXP$.  This follows
by using the $\PCP$ analogue of the classic $\textsc{MIP} = \NEXP$
result \cite{BFL91}.  We sharpen the parameters by using more modern
$\PCP$ machinery and then observing that the $\PCP$s we use can be
turned in to $\MIP$ at no cost.  In particular using the recent
results by Chan \cite{Chan12}, we get
\begin{theorem}
    \label{thm:nexp_bound}
    For every $\epsilon > 0 $ and $s = 2^{\lceil \log {(k+1)}\rceil}/2^k+\epsilon \le 2k/2^k + \epsilon$,
    $$\MIP[k, 1-\epsilon, s] = \NEXP.$$
\end{theorem}

Taken together, these results demonstrate that $k$-prover games
provide a natural ``quantitative'' way
to relate complexity classes such as $\BPP,\SZK,\AM,\EXP$ and $\NEXP$. 
We leave open the question of whether  $\MIP[k, 1-\epsilon, s]$
contains an even more fine grained hierarchy of complexity classes in
the regime where $s \geq 2/2^k + \epsilon$.

\subsection{Related work}
The work most closely related to our is the work by
Goldreich, Vadhan and Wigderson \cite{GVW02} mentioned above which in turn builds on a
work by Goldreich and H\aa stad \cite{GH98}; just as we do, both these
works investigate the complexity of interactive proofs with
``laconic'' provers.  We have taken the question to an extreme in one
direction (namely we focus only on provers that send a single bit); on
the other hand, we have generalized the question by considering
multi-prover interactive proofs, rather than just a single prover (as
is the main focus in the above-mentioned works).

The large literature on $\PCP$ characterizations of $\NP$ (e.g.,
\cite{AS98,ALMSS98,BGLR94,BGS98,GLST98,ST00} and many others) is clearly also very related. As mentioned, a $k$-prover $\MIP$
yields a $k$-query $\PCP$ with the same soundness error, but of
exponential length; typically, the $\PCP$ literature focuses on
polynomial-length proofs. Nonetheless, we rely on both $\PCP$s and
techniques from this literature (most notably Fourier analysis) to
analyze our proof system. 

We also mention the recent work by Drucker \cite{Dru11} that provides a 
$\PCP$-type characterization of $\AM$; his result is incomparable
to our main theorem as he focuses on polynomial-length PCP proofs.

\subsection{Outline}

In Section~\ref{sec:preliminaries} we present some definitions and
background material that we use.  In Section~\ref{sec:szk} we prove
Theorem~\ref{thm:szk_range} for the $\SZK$ range.  Our main result
Theorem~\ref{main.thm} is proved in Section~\ref{sec:am}.  The Theorems~\ref{thm:exp_bound} and \ref{thm:nexp_bound} are proved in Section~\ref{sec:exp}.
Finally, we end with discussing some avenues for future work in
Section~\ref{sec:conclusion}.

  \section{Preliminaries}
  \label{sec:preliminaries}

  \subsection{Laconic Proof systems}

  We assume familiarity with multi-prover interactive proofs and
  probabilistically checkable proofs.

  \begin{definition}
    $\IP[k,c,s]$ denotes the class of problems having an  two message
protocol where the first message is sent by the Verifier and where
the prover sends at most $k$ bits and where the proof has
soundness $s$ and completeness $c$.
  \end{definition}

  \begin{definition}
    $\MIP[k,c,s]$ denotes the set of languages having a Multi-prover
    Interactive Proof System with $k$ provers, each sending a single
    bit, soundness $s$, completeness $c$.  The questions to the
    $k$ provers are asked simultaneously.  In other words, all questions
    are formulated before any answer is recieved.
  \end{definition}
    
  \begin{fact}
    For every $k \ge 1$, $0 \le s < c \le 1$, we have
    $$\IP[k,c,s] \subseteq \MIP[k,c,s] .$$
  \end{fact}

When constructing $\MIP$ it is convenient to rely on efficient
$\PCP$s.  There are general translations from $\PCP$s to $\MIP$s
(one is given in \cite{BGS98}) if one accepts a slight loss
in the parameters.  In the cases we are interested
in, however, by a slight extension of the analysis we can turn the $\PCP$
directly into a $\MIP$ without any loss in parameters.

  \subsection{Statistical Zero Knowledge}

  For our characterization of the $\SZK$ range, we only need to rely
  on the following result of \cite{GVW02} relating $\SZK$ to laconic
  $\IP$ systems.

  \begin{theorem}[\cite{GVW02}, Theorem 3.1]
    \label{thm:ip=szk}
    For every $c$, $s$ such that $1 > c^2 > s > c/2 > 0$, it holds that
    $\IP[1, c, s] = \SZK$.
  \end{theorem}

  \subsection{Fourier Analysis of Boolean Functions}

  For two vectors $x, y \in \{0,1\}^n$ we write $x \oplus y$ for their
  pointwise sum modulo $2$.  Given $a \in \bool^n$
  we write $\chi_{a}: \bool^n \rightarrow \R$ for the \emph{character}
(which is in fact a linear function) $\chi_{a}(x) = (-1)^{\sum_{i=1}^n a_i x_i}$.

  Any Boolean function $f: \{0,1\}^n \rightarrow \R$ can be uniquely
  decomposed as a linear combination of characters
  \[
  f(x) = \sum_{a \in \bool^n} \hat{f}(a) \chi_a(x),
  \]
  where $\hat{f}(a) = \E_{x}[f(x) \chi_a(x)]$ are
  the \emph{Fourier coefficients} of $f$.

  We recall Plancherel's equality: for any $f: \bool^n \rightarrow \R$, we have
  $$\sum_{a} \hat{f}(a)^2 = \E_{x}[f(x)^2].$$

  \subsection{Inapproximability of Linear Equations}

  Our proof system for $\AM$ is based on the optimal inapproximability
  result for linear equations mod $2$ by H\aa stad \cite{Has01},
  defined next.

  \begin{definition}
    An instance $\Psi$ of \MaxLin{} consists of a set of equations in
    $n$ variables $x_1, \ldots, x_n$ over $\bool$.  Each equation is
    of the form $\chi_l(x) = b$ for some $l \in \{0,1\}^n$ of weight
    $3$ and some $b \in \{-1,1\}$.  We denote by $\Opt(\Psi) \in
    [0,1]$ the maximum fraction of equations satisfied by any
    assignment to $x$.
  \end{definition}

  \begin{theorem}[\cite{Has01}]
    For every $\epsilon > 0$, given a \MaxLin{} instance $\Psi$, it is
    NP-hard to determine whether $\Opt(\Psi) \le 1-\epsilon$ or
    whether $\Opt(\Psi) \ge \frac{1 + \epsilon}{2}$.
  \end{theorem}

  \section{The $\SZK$ range}
  \label{sec:szk}

  \begin{theorem}
    For every $k \ge 1$, $\epsilon > 0$, we have
    $$\IP[k,1-\epsilon,1/2^k+\epsilon] \supseteq \SZK.$$
  \end{theorem}

  \begin{proof}
    Follows by repetition of the protocol from
    Theorem~\ref{thm:ip=szk} and the fact that
    there is no problem with parallel repetition for
    one-prover proof systems.
  \end{proof}

  \begin{proposition}
    \label{prop:mip2ip}
    For every $k \ge 1$, $0 \le s \le c \le 1$, we have
    $$\MIP[k, c, s] \subseteq \IP[1, c, 2^{k-1}s].$$
  \end{proposition}

  \begin{proof}
    Given a $\MIP$ protocol $(V, P_1, \ldots, P_k)$ for a
    language $L$, we construct a single-prover protocol $(V', P')$ as
    follows.  The verifier $V'$ runs $V$ to generate $k$ messages
    $x_1, \ldots, x_k$, and sends $x_1$ to the prover $P'$.  The
    prover $P'$ acts as $P_1$ and responds with an answer $y_1 \in
    \{0,1\}$. $V'$ accepts iff there are bits $y_2, \ldots, y_k$ such
    that the original verifier $V$ accepts on the transcript $(x_1,
    \ldots, x_k, y_1, \ldots, y_k)$.  Clearly, the completeness of
    $(V', P')$ is at least that of the original protocol.  For the
    soundness, suppose that there is a strategy for $P'$ that makes
    the verifier accept with probability $s'$.  Construct a strategy
    for the original protocol by having $P_1$ act as $P'$ and $P_2,
    \ldots, P_k$ return random answers.  Clearly, these provers make
    $V$ accept with probability at least $s'/2^{k-1}$, implying $s'
    \le 2^{k-1}s$ as desired.
  \end{proof}

  \begin{theorem}
    \label{thm:mip-in-szk}
    For every $k \ge 1$, and every $\epsilon > 0$ it holds that
    $$\MIP[k, 1-\epsilon, 2/2^k(1-2\epsilon)] \subseteq \SZK$$
  \end{theorem}

  \begin{proof}
    We have
    \begin{align*}
      \MIP[ k, 1-\epsilon, 2/2^k(1-2\epsilon)] &\subseteq \IP[1,1, 1-\epsilon, 1-2\epsilon]  & \text{(Proposition~\ref{prop:mip2ip})} \\
      & \subseteq \SZK & \text{(Theorem~\ref{thm:ip=szk})}
    \end{align*}
  \end{proof}

    \section{Proof systems for $\AM$}
    \label{sec:am}

    First we note that, at a cost of an arbitrarily small loss in
    soundness and completeness, we may restrict ourselves to proof
    systems for $\NP$.
    
    \begin{lemma}
      If $\NP \subseteq \MIP[k,c,s]$ then for every $\epsilon > 0$ it holds that $\AM \subseteq \MIP[k,c-\epsilon,s+\epsilon]$
    \end{lemma}
    
    \begin{proof}
      Let $L \in \AM$.  We remind the reader that this is equivalent
      to the existence of a language $L' \in \NP$ such that $x\in L$ iff $(x,r)\in L'$
      with high probability for a random string $r$ (of an appropriate
      polynomial length).  Without loss of generality, we may assume that
      the protocol for $L$ has completeness $1-\epsilon$ and soundness
      $\epsilon$.  The $\MIP$ verifier for $L$ simply sends Arthur's
      random string $r$ to each of the $k$ provers and then executes
      the $\MIP$ protocol assumed to exist for $L' \in \NP$.

      If $x \in L$ then with probability $1-\epsilon$ over $r$ we have
      $(x,r) \in L'$ in which case the provers convince the verifier
      with probability $\ge c$.

      On the other hand $x \not\in L$ then the probability that the
      provers accept is at most $\Pr_r[(x,r) \in L'] + \Pr[(x,r) \not \in
        L']\Pr[\text{accept}\,|\,(x,r) \not \in L'] \le \epsilon + s$.
    \end{proof}

  \subsection{Warm-up: the case of $2$ provers}

  We start off with the case of only $2$ provers, as this case is
  somewhat simpler than the general case, and will be used to obtain
  the general case.  

  \begin{theorem}
    For every $\epsilon > 0$
    $$\NP \subseteq \MIP[2,1-\epsilon, 1/2+\epsilon].$$
  \end{theorem}

  \begin{proof}
    We reduce from the \MaxLin{} problem.  Given is a $\MaxLin{}$
    instance $\Psi$, on $n$ variables $x_1, \ldots, x_n$ and $m$
    linear equations $\{l_i(x_i) = b_i\}_{i \in [m]}$.

    The provers are expected to provide oracle access to the Hadamard
    encoding of a $(1-\epsilon)$-satisfying assignment $x \in
    \bool^n$.  In other words, the verifier will give each prover a
    vector $a \in \bool^n$ and expects in response the value of the
    linear function $\chi_a(x) \in \{-1,1\}$.
    
    The verifier proceeds as follows:
    \begin{enumerate}
    \item Pick a random equation $\chi_{l}(x) = b$ in $\Psi$
    \item Pick random $y \in \bool^n$ 
    \item Check that $P_2(y) \cdot P_1(l \oplus y) = b$
    \end{enumerate}

    It is easy to see that there is a strategy for the provers which
    makes the verifier accept with probability at least $\Opt(\Psi)$.
    More interestingly, we will now prove that, $\Opt(\Psi)$ is
    \emph{exactly} the maximum acceptance probability, over any
    strategy for $P_1$ and $P_2$.

    We can then write the acceptance probability of the verifier as
    \begin{eqnarray}\label{eq1}
      \Pr[\textrm{Verifier accepts}]  & = & \E_{\substack{y \in \bool^n\\(l,b) \in \Psi}}\left[\frac{1 + b P_1(l \oplus y)P_2(y)}{2}\right].
\end{eqnarray}
Replacing the two functions by their Fourier expansion we need to
analyze
$$
\sum_{a, a'} \hat P_1(a) \hat P_2(a') \E_{y,(l,b)}[
b \chi_{a} (l\oplus y) \chi_{a'} (y)].$$
All terms with $a\not= a'$ have expectation 0 and furthermore
we have 
$$
\left|\E_{(l,b)} [b \chi_a(l)]\right|  \leq 2 \Opt(\Psi) -1,
$$
as the assignment given by $a$ satisfies at most an $\Opt(\Psi)$
fraction of the equations and at least a fraction $1-\Opt(\Psi)$
as its negation does not satisfy more than a $\Opt(\Psi)$ fraction.
We conclude that (\ref{eq1}) is bounded by
$$
\frac {1+\sum_a |\hat P_1(a) \hat P_2(a)| (2 \Opt(\Psi) -1)}2.
$$
Finally note that, by Cauchy-Schwarz,
$$
\sum_a |\hat P_1(a) \hat P_2(a)|
\leq \left( \sum_a \hat P^2_1(a)\right)^{1/2}
\left( \sum_a \hat P^2_2(a) \right)^{1/2} =1
$$
and this finishes the argument.
\end{proof}

  \subsection{The general case}

We have

  \begin{theorem}
    For every $k \ge 2$, $\epsilon > 0$, we have
    $$\NP \subseteq \MIP[k,1-\epsilon, 2/2^k+\epsilon].$$
  \end{theorem}

  \begin{proof}
    As before, we design a $\MIP$ system for linear equations.  Given is
    a $\MaxLin{}$ instance $\Psi$, in which either $\Opt(\Psi) \ge
    1-\epsilon_0$, or $\Opt(\Psi) \le \frac{1 + \epsilon_0}{2}$, where
    $\epsilon_0$ will be chosen small enough to get the completeness
    and soundness bound that we want.

    The verifier again expects all the $k$ provers to provide answers
    to the Hadamard coding of the good assignment, and it then does the obvious
    generalization of the $k=2$ case:
    \begin{enumerate}
    \item Pick $k-1$ random equations $l_j(x) = b_j$, $1 \le j \le k-1$
    \item Pick random $y \in \bool^n$
    \item Check that $P_j(l_j \oplus y) \cdot P_k(y) = b_j$ for
      every $1 \le j \le k-1$
    \end{enumerate}

    It is clear that the completeness is at least $(1-\epsilon_0)^{k-1}
    \ge 1-k\epsilon_0$.  Thus, as long as $\epsilon_0 \le \epsilon/k$,
    we have the desired completeness.
    
    Let us now study the soundness, i.e., the maximum possible
    acceptance probability of verifier, given that $\Opt(\Psi) \le
    \frac{1+\epsilon_0}{2}$.
    
    We say that prover $P_j$ \emph{succeeds} if $P_j(l_j \oplus y)
    \cdot P_k(y) = b_j$.  From the analysis of the previous theorem,
    we know that the probability that $P_j$ succeeds is at most
    $\frac{1 + \epsilon_0}{2}$.  Thus, if the events that the
    different provers succeed were independent, we would obtain the
    desired soundness of $\approx 2^{1-k}$.  However, {\em a priori}, it may
    be that the success events of the provers are very correlated,
    e.g., it could be that if one succeeds then they all succeed.

    To cope with this, we need to obtain a more robust version of the
    previous analysis.  Let $\frac {1+\delta_j(y)}2$ be the 
    probability that $P_j$ succeeds given that $y$ is chosen.  We
    have the following lemma.

\begin{lemma}\label{lemma:dy2}
    $\E_y[\delta_j^2 (y)] \leq \epsilon_0^2.$
\end{lemma}

\begin{proof}
We have
$\delta_j(y)= \E_{(l,b)}[ b P_k(y)P_j(l \oplus y)]$ and
thus
$$\E_y[\delta_j^2 (y)] =
\E_{(l,b),(l',b'),y} [bb'P_j (l\oplus y) P_j (l'\oplus y)].$$
Similarly to the case $k=2$ we replace the function by its
Fourier expansion and we are left to analyze
$$
\sum_{a,a'} \hat P_j(a)\hat P_j (a') \E_{y,(l,b),(l',b')} [
bb'\chi_a(l+y) \chi_{a'}(l'+y) ].
$$
Again we only have nonzero terms when $a=a'$.  For these terms it
easy to see that
$$
\left| \E_{(l,b),(l',b')} [bb'\chi_a(l) \chi_{a}(l') ] \right| \leq
(2\Opt (\Psi)-1)^2= \epsilon_0^2.$$
Using $\sum_a \hat P_j(a)^2=1$, the lemma follows.
\end{proof}

Lemma \ref{lemma:dy2} implies that the fraction of $y$ such that 
$\delta_j(y) \geq \sqrt \epsilon_0$ is bounded by $\epsilon_0$.

We conclude that the, for the $y$ chosen, the probability
that $\delta_j(y) \geq \sqrt \epsilon_0$ for any
$j$ is bounded by $k\epsilon_0$.  On the other hand if 
$\delta_j(y) \leq \sqrt \epsilon_0$ for all values of $j$ the probability
that the verifier accepts is bounded by $(\frac {1+\sqrt {\epsilon_o}}2)^{k-1}$.
We conclude that the overall probability that the verifier accepts
is bounded by
\[
k\epsilon_0+\left(\frac {1+\sqrt {\epsilon_o}}2\right)^{k-1},\]
and choosing $\epsilon_0$ sufficiently small, this is bounded
by $2^{1-k}+\epsilon$.
\end{proof}

\section{The High End -- $\EXP$ and $\NEXP$ Results}
\label{sec:exp}

In this section we prove Theorems~\ref{thm:exp_bound} and
\ref{thm:nexp_bound}.  These are essentially just ``blow-ups'' of
corresponding approximation algorithms and inapproximability results.

\begin{theorem}[Theorem \ref{thm:exp_bound} restated]
   For all sufficiently large $k$, $\epsilon > 0$, $s \leq \frac{0.62k}{2^k}(1 - \epsilon)$ we have
   $$\MIP[k, 1-\epsilon, s] \subseteq \EXP.$$
   This holds also for $k = 3$ and $s \le 1/2-\epsilon$.
\end{theorem}

\begin{proof}
  Let $L \in \MIP[k, 1-\epsilon, s]$ with $s \le \frac{0.62k}{2^k}(1-\epsilon)$.  Given
  an instance, the task of determining whether $x \in L$ boils down to
  finding the best joint strategy for the $k$ provers.  If the
  verifier uses $r$ random bits she can send at most $2^r$ different
  queries to each prover, thus the optimal strategy can be described
  by $k \cdot 2^r = 2^{\poly |x|}$ bits.  Further, for each outcome of the verifier's
  randomness, the acceptance criterion is a constraint on some $k$
  bits of the strategy.  In other words, what we have is an
  exponentially large Max $k$-CSP instance.  The value of this
  instance can be approximated in time polynomial in its size to
  within a factor $0.62k/2^k$ by the algorithm of Makarychev and Makarychev~\cite{MM12}.  For the case $k = 3$ we use the $1/2$-approximate
  Max $3$-CSP algorithm of Zwick \cite{Zwick98}.
\end{proof}

Next we show that if the soundness is sufficiently large,
exponential-size $k$-query $\PCP$ systems can express every language
in $\NEXP$.

\begin{theorem}
  \label{thm:nexp_pcp}
  For $t = 2^{\lceil \log_2(k+1) \rceil}$ ($k+1$ rounded up to the next power of $2$) we have
  $$\MIP[ k, 1-\epsilon, t/2^k+\epsilon] = \NEXP.$$
\end{theorem}

This immediately implies Theorem~\ref{thm:nexp_bound}.

\begin{proof}[Proof sketch]
The proof follows from a upscaling of the
recent $\PCP$ of Chan \cite{Chan12} that gives a
predicate of arity $k$ which has $t$ accepting
configurations and which is approximation resistant.

In a standard $\PCP$, the verifier runs in polynomial
time, uses a logarithmic number of random coins
and reads a constant number of bits in a polynomial
size proof and verifies an NP-statement.  We are currently
interested in the situation where the crucial parameters,
except the running time of the verifier, are exponentially
larger.  

To be more precise we are interested in
a polynomial time verifier, that uses a polynomial number
of random coins and gets one bit
each from $k$ different provers that respond to questions of
polynomial length.

As is convenient for us, Chan already analyzed his $\PCP$
in the $k$-partite situation where each bit is read
from a separate table.  This model is exactly the same
as a $k$-prover model and hence this difference is
only syntactical.

It remains to address the question on how to make the upscaled
verifier run in polynomial time.  This amounts to saying
that a verifier of an $\NEXP$ statement
runs in polynomial time.  This was 
explicitly needed in \cite{BFL91} but this paper
predates the $\PCP$-Theorem.  The fact that this
is true also for upscaled versions of the $\PCP$-Theorem has been explicitly
stated in \cite{BGS98} and \cite{BGHSV05}.  The intuitive reason
that this is true is that the verifier only needs to ensure
that some bits in a suitable encoding of the inputs are
correct and this takes polynomial time in the size of
the input but not the other parameters of the proof.
\end{proof}

\section{Concluding Remarks}
\label{sec:conclusion}

There are a number of interesting avenues for further work.  In this
paper we focused solely on the case of almost perfect completeness and
each prover sending exactly $1$ bit.  Obviously, understanding what
happens with the expressiveness of these systems for other
completeness values (in particular perfect completeness) and slightly
less laconic provers would be very interesting.  By simple extensions
of the methods used in this paper it is possible to get some results
but it would be interesting to see if perfect completeness could
lead to a significantly different situation in any range
of parameters.

There is also a specific question more directly related
to the current paper.
  There is a huge gap between our lower bound $\AM$ and upper bound
  $\EXP$ for soundness $s = 2/2^k + \epsilon$.  It seems quite
  plausible that an upper bound for this range of $s$ should be
  $\PSPACE$ rather than $\EXP$ -- proving this essentially boils down
  to proving that there is a $\delta > 0$ such that bipartite
  instances of Max $2$-CSP can be approximated within a factor
  $1/2+\delta$ in polylog-space (and not necessarily polynomial time).
  We hope that the recent algorithms for Max Cut, in particular
  \cite{KS11}, can be adapted to achieve this.
  
  Even if this turns out to be true, whether the correct class
  here is $\AM$ or $\PSPACE$ or something in between we have little
  intuition about.

\smallskip
\noindent
{\bf Acknowledgment.}  We are grateful to Salil Vadhan for
pointing out a simple proof of the lower bound given in
Theorem~\ref{thm:szk_range} rather than the more complicated
proof with worse parameters that we originally had.
We are also grateful to Madhu Sudan and Or Meir for discussions
on how to blow-up PCPs.

\bibliographystyle{alpha} \bibliography{laconic}
  
\end{document}